\def\year{2023}\relax

\newif\iflong
\newif\ifshort

\longtrue

\iflong
\else
\shorttrue
\fi

\documentclass[letterpaper]{article} \usepackage{aaai23}  \usepackage{times}  \usepackage{helvet}  \usepackage{courier}  \usepackage[hyphens]{url}  \usepackage{graphicx}    \usepackage{natbib}  \usepackage{caption} \DeclareCaptionStyle{ruled}{labelfont=normalfont,labelsep=colon,strut=off}       \usepackage{algorithm}
\usepackage{algorithmic}

\usepackage{newfloat}
\usepackage{listings}
\floatstyle{ruled}
\newfloat{listing}{tb}{lst}{}
\floatname{listing}{Listing}
\usepackage{thm-restate}
\usepackage{amsmath}
\usepackage{amssymb}
\usepackage{amsthm}
\usepackage{xspace}
\usepackage{bm}
\usepackage{complexity}
\usepackage[textwidth=16mm,textsize=footnotesize]{todonotes}
\usepackage{tikz}
\usetikzlibrary{shapes,fit,hobby,backgrounds,positioning,calc,trees,decorations.pathreplacing}
\usepackage{boxedminipage}
\usepackage{cleveref}

\usepackage{booktabs}

\newcommand{\N}{\ensuremath{\mathbb{N}}}
\newcommand{\Nat}{\N}

\newcommand{\td}{\ensuremath{\operatorname{td}}\xspace}  

\newcommand{\bigoh}{\ensuremath{\mathcal{O}}}

\newcommand{\III}{\mathcal{I}}

\newcommand{\SSS}{\mathcal{S}}

\newcommand{\XXX}{\mathcal{X}}

\newtheorem{theorem}{Theorem}
\crefname{theorem}{Theorem}{Theorems}

\crefname{observation}{Observation}{Observations}
\newtheorem{lemma}[theorem]{Lemma}
\crefname{lemma}{Lemma}{Lemmas}
\newtheorem{corollary}[theorem]{Corollary}
\crefname{corollary}{Corollary}{Corollaries}

\crefname{proposition}{Proposition}{Propositions}

\crefname{conjecture}{Conjecture}{Conjectures}
\newtheorem{claim}{Claim}
\newtheorem{definition}[theorem]{Definition}
\crefname{claim}{Claim}{Claims}
\newenvironment{claimproof}[1]{\par\noindent\emph{Proof of the Claim.}\space#1}{\hfill $\blacksquare$}
\newenvironment{claimproofsketch}[1]{\par\noindent\emph{Proof of the Claim (Sketch).}\space#1}{\hfill $\blacksquare$}

\theoremstyle{remark}

\crefname{example}{Example}{Examples}

\newcommand{\dom}{\ensuremath{\mathbb{B}}\xspace}
\newcommand{\innbhd}[1]{\ensuremath{\delta^-(#1)}}
\newcommand{\SDSys}{SyDS\xspace}
\newcommand{\reach}{\textsc{Reachability}}
\newcommand{\conv}{\textsc{Convergence}}
\newcommand{\QBF}{\textsc{QBF}}
\newcommand{\allconv}{\textsc{Convergence Guarantee}}

\newcommand{\pervec}{\operatorname{pv}}

\begin{document}

\title{
A Structural Complexity Analysis of Synchronous Dynamical Systems
}
\author {
    Eduard Eiben\textsuperscript{\rm 1}, Robert Ganian\textsuperscript{\rm 2}, Thekla Hamm\textsuperscript{\rm 3}, Viktoriia Korchemna\textsuperscript{\rm 2}    
}
\affiliations {
    \textsuperscript{\rm 1} Department of Computer Science, Royal Holloway, University of London, UK \\
     \textsuperscript{\rm 2} Algorithms and Complexity Group, TU Wien, Austria \\
    \textsuperscript{\rm 3} Eindhoven University of Technology, the Netherlands \\
    Eduard.Eiben@rhul.ac.uk, rganian@gmail.com, vkorchemna@ac.wien.ac.at,  thekla.hamm@gmail.com
}

\maketitle

\begin{abstract}
Synchronous dynamic systems are well-established models that have been used to capture a range of phenomena in networks, including opinion diffusion, spread of disease and product adoption.
We study the three most notable problems in synchronous dynamic systems: whether the system will transition to a target configuration from a starting configuration, whether the system will reach convergence from a starting configuration, and whether the system is guaranteed to converge from every possible starting configuration. While all three problems were known to be intractable in the classical sense, we initiate the study of their exact boundaries of tractability from the perspective of structural parameters of the network by making use of the more fine-grained parameterized complexity paradigm. 

As our first result, we consider treewidth---as the most prominent and ubiquitous structural parameter---and show that all three problems remain intractable even on instances of constant treewidth. We complement this negative finding with fixed-parameter algorithms for the former two problems parameterized by treedepth, a well-studied restriction of treewidth. While it is possible to rule out a similar algorithm for convergence guarantee under treedepth, we conclude with a fixed-parameter algorithm for this last problem when parameterized by treedepth and the maximum in-degree.
\end{abstract}

\section{Introduction}

Synchronous dynamic systems are a well-studied model used to capture a range of diffusion phenomena in networks~\cite{RosenkrantzMRS21,ChistikovLPT20}.
Such systems have been used, e.g., in the context of social contagions (e.g., the spread of information, opinions, fads, epidemics) as well as product adoption~\cite{AdigaKMMRV18,GuptaSMVLL18,OgiharaU17}. 

Informally, a synchronous dynamic system (SyDS) consists of a directed graph $G$ (representing an underlying network) with each node $v$ having a local function $f_v$ and containing a state value from a domain $\mathbb{B}$, which may evolve over discrete time steps\footnote{Formal definitions are provided in the Preliminaries.}. While each node $v$ begins with an initial state value at time $0$, at each subsequent time step it receives an updated value by invoking the node's local function $f_v$ on the state value of $v$ and of all nodes with arcs to $v$ (i.e., the closed in-neighborhood of $v$).
In line with recent works~\cite{RosenkrantzMRS21,ChistikovLPT20}, here we focus our attention to the Boolean-domain case with deterministic functions, which is already sufficiently rich to model a variety of situations. SyDS with Boolean domains are sometimes also called synchronous Boolean networks, especially in the context of systems biology~\cite{OgiharaU20,AkutsuCN07,KaufmannPST03}.

A central notion in the context of SyDS is that of a \emph{configuration}, which is a tuple specifying the state of each node at a certain time step. In several use cases of SyDS, there are clearly identifiable configurations that are either highly desirable (e.g., when dealing with information propagation), or highly undesirable (when modeling the spread of a disease or computer virus). Indeed, the \reach\ problem~\cite{RosenkrantzMRS21}---deciding whether a given target configuration will be reached from a given starting configuration---is a classical computational problem on SyDS~\cite{OgiharaU17,AkutsuCN07}. 

In other settings such as in opinion diffusion~\cite{AulettaFG18,AulettaCFGP17}, we do not ask for the reachability of a specific configuration, but rather whether the system eventually converges into a \emph{fixed point}, i.e., a configuration that transitions into itself. This idea has led to the study of two different problems on SyDS~\cite{ChistikovLPT20}: in \conv\ we ask whether the system converges (to an arbitrary fixed point) from a given starting configuration, while in \allconv\ we ask for a much stronger property---notably, whether the system converges from all possible configurations. 

In view of the fundamental nature of these three problems, it is somewhat surprising that so little is known about their complexity. The \PSPACE-completeness of \conv\ and \allconv\ has been established two years ago~\cite{ChistikovLPT20}, while the \PSPACE-completeness of \reach\ on directed acyclic networks was established even more recently~\cite{RosenkrantzMRS21}. Earlier, Barrett et al.~\cite{BarrettHMRRS06} established the \PSPACE-completeness of \reach\ on general directed networks of bounded treewidth and degree, albeit the bounds obtained in that work are very large.
In spite of these advances, we still lack a detailed understanding of the complexity of fundamental problems on SyDS.

\smallskip
\noindent \textbf{Contribution.}\quad
Since \reach, \conv\ and \allconv\ are all computationally intractable on general SyDS, it is natural to ask whether this barrier can be overcome by exploiting the structural properties of the input network. In this paper, we investigate these three problems through the lens of \emph{parameterized complexity}~\cite{DowneyF13,CyganFKLMPPS15}, a computational paradigm which offers a refined view into the complexity-theoretic behavior of problems. In this setting, we associate each input $\mathcal{I}$ with a numerical parameter $k$ that captures a certain property of $\mathcal{I}$, and ask whether there is an algorithm that can solve such inputs in time $f(k)\cdot |\mathcal{I}|^{\bigoh(1)}$ (for some computable function $f$)---if yes, the problem is called \emph{fixed-parameter tractable}, and the class \FPT\ of all such problems is the central notion of tractability in the parameterized setting. 

We begin our investigation by considering the most widely studied and prominent graph parameter, notably \emph{treewidth}~\cite{RobertsonS84}. Treewidth intuitively captures how ``tree-like'' a network is, and it is worth noting that in addition to its fundamental nature, the structure of real-world networks has been demonstrated to attain low treewidth in several settings~\cite{ManiuSJ19}. While \allconv\ was already known to be intractable even on networks of constant treewidth~\cite[Theorem 5.1]{RosenkrantzMRS21}\footnote{In fact, the problem can be shown to be intractable even on stars via a simple argument.}, previous reductions for \reach\ and \conv\ only apply to networks of high treewidth~\cite{ChistikovLPT20}. Here, we show: 

\ifshort
\begin{restatable}{theorem}{thmpspace}
\label{thm: pspace}
\reach\ and \conv\ are \PSPACE-complete, even on \SDSys{s} of treewidth \(2\) and maximum in-degree \(3\).
\end{restatable}
\fi

\iflong
\begin{restatable}{theorem}{thmpspace}
\label{thm: pspace}
\reach\ and \conv\ are \PSPACE-complete, even on \SDSys{s} of treewidth \(2\) and maximum in-degree \(3\).
\end{restatable}
\fi

The main technical contribution within the proof is the construction of a non-trivial counter which can loop over all exponentially many configurations of a set of nodes and whose structure is nothing more than a directed path. We believe the existence of such a counter is surprising and may be of independent interest; it contrasts previous counter constructions which relied on much denser connections between the nodes, but interestingly its simple structure comes at the cost of the configuration loop generated by the counter being highly opaque.

Intractability w.r.t.\ treewidth draws a parallel to the complexity behavior of the classical \QBF\ problem---an archetypical \PSPACE-complete problem which also remains \PSPACE-complete on instances with bounded treewidth~\cite{AtseriasO14}. In fact, while being based on entirely different ideas, our reduction and Atserias' and Oliva's construction for \QBF\ also show intractability for a slight restriction of treewidth called \emph{pathwidth}, but do not exclude tractability under a related parameter \emph{treedepth}~\cite{sparsity}. Investigating the complexity of our problems under the parameter treedepth is the natural next choice, not only because it lies at the very boundary of intractability, but also because of its successful applications for a variety of other problems~\cite{GanianPSS20,GanianO18,GutinJW16} and its close connection to the maximum path length in the network\footnote{A class of networks has bounded treedepth if and only if there is a bound on the length of any undirected path.}.
While the complexity of \QBF\ parameterized by treedepth remains a prominent open problem, as our second main technical contribution we show:

\begin{restatable}{theorem}{thmtreedepth}
\label{thm: treedepth}
\reach\ and \conv\ are fixed-parameter tractable when parameterized by the treedepth of the network. 
\end{restatable}

The main idea behind the proof of Theorem~\ref{thm: treedepth} is to show that the total periodicity of the configurations is upper-bounded by a function of the treedepth, and this fact then enables us to argue the correctness of an iterative pruning step that allows us to gradually reduce the instance to an equivalent one of bounded size. As for the third problem (\allconv), fixed-parameter tractability w.r.t.\ treedepth is excluded by the intractability of the problem on stars.

While these results already provide a fairly tight understanding of the complexity landscape for two out of the three studied problems, they do raise the question of which structural properties of the network can guarantee the tractability of \allconv. Intuitively, one of the main difficulties when dealing with \allconv\ is that it is not even possible to enumerate all possible starting configurations of the network. Yet, in spite of this seemingly critical problem, we conclude by establishing fixed-parameter tractability of \allconv\ when parameterized by treedepth plus the maximum in-degree:

\begin{restatable}{theorem}{thmtreedepthindeg}
\label{thm: treedepth_indeg}
\allconv\ is fixed-parameter tractable when parameterized by the treedepth plus the maximum in-degree of the network. 
\end{restatable}

Our results are summarized in Table~\ref{tab:results}. For completeness, we remark that the results are robust in terms of the type of functions that can be used---in particular, all our algorithmic results hold even if we assume that the functions are black-box oracles.
Moreover, in some settings it may be useful to ask whether a target configuration and/or convergence is reached up to an input-specified time point; incorporating this as an additional constraint leads to a strict generalization of \reach, \conv\ and \allconv, and all of the algorithms provided here can also directly solve these more general problems.

\begin{table*}[ht]
  \centering
  \begin{tabular}{ccccc}
    \toprule
    ~ & Unrestricted & Treewidth & Treedepth & Treedepth + in-degree\\ \midrule
    \reach & \PSPACE-c~\cite{RosenkrantzMRS21} & $\boldsymbol{\mathsf{PSPACE}}$\textbf{-c}$^\dagger$ & \boldsymbol{\mathsf{FPT}} & \FPT \\
    \conv & \PSPACE-c~\cite{ChistikovLPT20} & $\boldsymbol{\mathsf{PSPACE}}$\textbf{-c} & \boldsymbol{\mathsf{FPT}} & \FPT \\
    \allconv & \PSPACE-c~\cite{ChistikovLPT20} & co\NP-h& co\NP-h & \boldsymbol{\mathsf{FPT}}\\
    \bottomrule
\end{tabular}
\caption{Summary of our main results (marked in bold). These include (1) the \PSPACE-completeness of the former two problems on networks of constant treewidth, (2) their fixed-parameter tractability with respect to the parameter treedepth, and (3) a fixed-parameter algorithm for \allconv\ when parameterized by treedepth plus the in-degree of the network. The co\NP-hardness of the latter problem on networks of constant treedepth and treewidth follows from previous work~\cite[Theorem 5.1]{RosenkrantzMRS21}).\\
$^\dagger$ The \PSPACE-completeness of \reach\ on inputs of bounded treewidth was already shown by Barrett et al.~\cite{BarrettHMRRS06}, albeit the constants used in that reduction were very large while here we establish intractability for treewidth~$2$.
}

\label{tab:results}
\end{table*}

\section{Preliminaries}
We use standard terminology and notation for directed simple graphs~\cite{Diestel12}, which will serve as models for the networks considered in this paper. 
We use \(\innbhd{v}\) to denote the \emph{in-neighbourhood} of a node \(v\) in a directed graph, i.e., the set of all nodes \(w\) such that the graph contains an edge \(wv\) which starts in \(w\) and ends in \(v\). 

It will be useful for us to consider tuples as implicitly indexed.
This means, for two sets \(A\) and \(B\) we use \(B^A\) to denote the set of tuples with \(|A|\)-many entries, each of which is an element of \(B\) \emph{and} at the same time we associate each of the entries with a fixed element of \(A\).
For a tuple \(x \in B^A\) and an element \(a \in A\), we denote by \(x_a\) the entry of \(x\) that is associated with \(a\).

\paragraph{Synchronous Dynamic Systems.}
A \emph{synchronous dynamic system} (\emph{\SDSys}) \(\SSS = (G,\dom,\{f_v \mid v \in V(G)\})\) consists of an \emph{underlying directed graph} (the \emph{network}) \(G\), a \emph{node state domain} \dom, and for each node \(v \in V(G)\) its \emph{local function} \(f_v \colon \dom^{\innbhd{v} \cup \{v\}} \to \dom\).
In line with previous literature, in this paper we will always consider \(\dom = \{0,1\}\) to be binary; however, all results presented herein can be straightforwardly generalized to any fixed (i.e., constant-size) domain.
A \emph{configuration} of a \SDSys is a tuple in \(\dom^{V(G)}\).
The \emph{successor} of a configuration \(x\) is the configuration \(y\) which is given by \(y_v = f_v(x_v)\) for every \(v \in V(G)\).
A configuration \(x\) is a \emph{fixed point} of a \SDSys if it is its own successor, i.e.\ for all \(v \in V(G)\), \(f_v(x_v) = x_v\).

Given two configurations \(x\) and \(y\) of a \SDSys \(\SSS\), we say \(y\) is \emph{reachable} from \(x\) if there is a sequence of configurations of \(\SSS\) starting in \(x\) and ending in \(y\) such that every configuration in the sequence is followed by its successor.

The three problems on synchronous dynamic systems considered in this paper are formalized as follows:
\begin{center}
	\begin{boxedminipage}{0.98 \columnwidth}
		\reach\\[5pt]
		\begin{tabular}{p{0.08 \columnwidth} p{0.8 \columnwidth}}
			Input: & A \SDSys \(\SSS = (G,\dom,\{f_v \mid v \in V(G)\})\), a configuration \(x\) of \(\SSS\) called \emph{starting configuration} and a configuration \(y\) of \(\SSS\) called \emph{target configuration}.\\
			Task: & Determine whether \(y\) is reachable from \(x\).
		\end{tabular}
	\end{boxedminipage}
\end{center}

\begin{center}
	\begin{boxedminipage}{0.98 \columnwidth}
		\conv\\[5pt]
		\begin{tabular}{p{0.08 \columnwidth} p{0.8 \columnwidth}}
			Input: & A \SDSys \(\SSS = (G,\dom,\{f_v \mid v \in V(G)\})\) and a configuration \(x\) of \(\SSS\) called \emph{starting configuration}.\\
			Task: & Determine whether there is a fixed point of \(\SSS\) that is reachable from \(x\).
		\end{tabular}
	\end{boxedminipage}
\end{center}

\begin{center}
	\begin{boxedminipage}{0.98 \columnwidth}
		\allconv\\[5pt]
		\begin{tabular}{p{0.08 \columnwidth} p{0.8 \columnwidth}}
			Input: & A \SDSys \(\SSS = (G,\dom,\{f_v \mid v \in V(G)\})\).\\
			Task: & Determine whether for every configuration \(x\) of \(\SSS\) there is a fixed point of \(\SSS\) that is reachable from \(x\).
		\end{tabular}
	\end{boxedminipage}
\end{center}

\paragraph{Treewidth and Treedepth.}
Similarly to other applications of treewidth on directed networks~\cite{GanianHO21,GuptaR0Z18a}, in this submission we consider the treewidth and treedepth of the \emph{underlying undirected graph}, which is the simple graph obtained by ignoring the orientations of all arcs in the graph. While directed analogues for treewidth have been considered in the literature, these have constant values on DAGs and hence cannot yield efficient algorithms for any of the considered problems~\cite{RosenkrantzMRS21}.

While treewidth has a rather technical definition involving bags and decompositions~\cite{RobertsonS84,DowneyF13}, for the purposes of this article it will suffice to remark that graphs where removing a single node results in a tree have treewidth at most $2$.

\smallskip
Treedepth is a parameter closely related to treewidth---in particular, the treedepth of a graph is lower-bounded by its treewidth. A useful way of thinking about graphs of bounded treedepth is that they are (sparse) graphs with no long paths. We formalize the parameter below.

A \emph{rooted forest} $\mathcal F$ is a disjoint union of rooted trees.
For a node~$x$ in a tree~$T$ 
of $\mathcal F$, the \emph{height} (or {\em depth})
of~$x$ in $\mathcal F$ is the number of nodes in the path from 
the root of~$T$ to~$x$.
The \emph{height of a rooted forest} is the maximum height of a node of the forest. 
Let $V(T)$ be the node set of any tree $T \in \mathcal F$.
\begin{definition}[Treedepth]\label{def:td}
  Let the \emph{closure} of a rooted forest~$\mathcal F$ be the graph
  $\lambda({\mathcal F})=(V_c,E_c)$ with the node set 
  $V_c=\bigcup_{T \in \mathcal F} V(T)$ and the edge set
  $E_c=\{xy \mid \text{$x$ is an ancestor of $y$ in some $T\in\mathcal F$}\}$.
  A \emph{treedepth decomposition}
  of a graph $G$ is a rooted forest $\mathcal F$ 
  such that $G \subseteq \lambda(\mathcal F)$.
  The \emph{treedepth} $\td(G)$ of a graph~$G$ is the minimum height of
  any treedepth decomposition of $G$.
\end{definition}

It is known that an optimal-width treedepth decomposition can be computed by a fixed-parameter algorithm~\cite{sparsity,ReidlRVS14,treedepthcomp} and also, e.g., via a SAT encoding~\cite{GanianLOS19}; hence, in our algorithms we assume that such a decomposition is provided on the input.

\section{The Path-Gadget and Hardness for Treewidth}
We provide a construction showing that even \SDSys{s} which are directed paths can reach an exponential number of configurations.
This functions as a crucial gadget for showing \PSPACE-hardness of \reach\ and \conv, but is also an interesting result in its own right as it significantly simplifies known constructions of \SDSys{s} with exponential periods~\cite{RosenkrantzMRS21}.

\iflong
\begin{theorem}
\fi
\ifshort
\begin{theorem}
\fi
\label{thm:paths}
For every $n\in \mathbb{N}$, there is a \SDSys \(\XXX^{(n)} = (G,\dom = \{0,1\},\{f_v \mid v \in V(G)\})\) and an initial configuration $x^0$, such that 
\begin{itemize}
\item $G$ is a directed path $v_1v_2\ldots v_{2n}$, 
\item $x^0$ is the all-zero configuration, 
\item the successor of $x^0$ is the configuration $x^1$ such that $x^1_{v_j}=1$ if and only if $j=1$ or $(j\mod 2)= 0$, and
\item for every $i\in [n+1]$ and every tuple $t\in \dom^{i}$, there exists $q\in \{0,\ldots, 2^i-1\}$ such that, for every $p\in \mathbb{N}$, the configuration of $\XXX^{(n)}$\ after $2^i\cdot p + q$ successor steps restricted to $v_1, v_2, v_4, \ldots, v_{2i-4}, v_{2i-2}$ is equal to $t$. 
\end{itemize}
In particular, for every $p\in \mathbb{N}$, after $2^i\cdot p$ steps, the node state of each of the nodes $v_1, v_2, v_4, \ldots, v_{2i-4}, v_{2i-2}$ is $0$.\end{theorem}

\iflong
\begin{proof}
\fi
\ifshort
\begin{proof}[Proof Sketch]
\fi
Let us fix some $n\in \mathbb{N}$. For simplicity of notation, we will denote the local function $f_{v_j}$ for the node $v_j$ by $f_j$.
 The local function for the node $v_1$ is $f_{1} \colon \dom \to \dom$ given by $f_1(b) = 1-b$, that is the configuration alternates between $0$ and $1$. For every $i\in [n]$, the local function for the node $v_{2i}$ is $f_{2i} \colon \dom^2 \to \dom$ given by $f_{2i}(b_1, b_2) = (b_1-b_2+1)\cdot (b_2-b_1+1)$. Equivalently, $f_{2i}(b_1, b_2) = 1$ if and only if $b_1=b_2$ and $f_{2i}(b_1, b_2) = 0$ otherwise, and so the function is simply an evaluation of the equivalence relation. 
 Next, for every $i\in [n-1]$, the local function for the node $v_{2i+1}$ is $f_{2i+1} \colon \dom^2 \to \dom$ given by $f_{2i+1}(b_1, b_2) =b_1\cdot (1-b_2)$. Equivalently, the configuration on $v_{2i+1}$ is $1$ if and only if in the previous step, the configuration on $v_{2i+1}$ was $0$ and the configuration on $v_{2i}$ was $1$. 

This finishes the description of the \SDSys. We will now prove that the structure of configurations over the time steps has the desired properties. To this end, we denote the initial configuration $x^0$ and for configuration $x^i$, we denote its successor as $x^{i+1}$. Moreover, for simplicity of the notation, we denote by $x^i_j$ the state of the node $v_j$ in $i$-th step, that is $x^i_{v_j}$. 

Clearly, $x^i_1 = 0$ if $i \mod 2 = 0$ and $x^i_1 = 1$ if $i \mod 2 = 1$. 
Moreover, it is also easy to verify now that $x^1_j=1$ if and only if $j=1$ or $(j \mod 2)=0$ since for $j\ge 2$, we have by the definition of the local functions that $f_j(0,0)= 1$ if and only if $j$ is even. 
In order to complete the proof, we first establish the following claim.

\iflong
\begin{claim}
\fi
\ifshort
\begin{claim}[$\star$]
\fi
For all $j\ge 1$ and $i\in \mathbb{N}$, it holds that $x^i_j = x^{i'}_j$, where $i' = i\mod 2^{\lfloor\frac{j}{2}\rfloor + 1}$. 
Moreover, if $j$ is even, then $x^i_j = 1-x^{i''}_j$, where $i'' = i+ 2^{\lfloor\frac{j}{2}\rfloor}$.
\end{claim}

\iflong
\begin{claimproof}
\fi
\ifshort
\begin{claimproofsketch}
\fi
Let us first introduce some notation that will help with exposition. For a node $v_j$, we will let the \emph{period vector} for the node $v_j$ be the vector $\pervec_j = (x^0_jx^1_j\ldots x^{2^{\lfloor\frac{j}{2}\rfloor + 1}-1}_j)$ of length $2^{\lfloor\frac{j}{2}\rfloor + 1}$. For simplicity, we mostly split the vector into smaller pieces consisting of at most three entries. Moreover, we use the exponent for a piece to specify how many times the same piece repeats in a row in the vector. For example, we could write the vector $(01001011)$ as $(010)^2(11)$.  
The reason we call these the ``period vectors'' is that we will show that for all $i\in \mathbb{N}$ we have $x^i_j = x^{i'}_j$, where $i' = i\mod 2^{\lfloor\frac{j}{2}\rfloor + 1}$; in other words, the period vector repeats cyclically as the state $x^{i}_j$ of the node $v_j$ starting from $x^0_j$. 

Let us proceed by computing the few first period vectors. This is straightforward, as we are always computing $\pervec_{j+1}$ from $\pervec_{j}$ using the fact that $x^0_{j+1}=0$. In each of the cases, we also make sure that the period vector $\pervec_{j}$ indeed repeats as the state changes of node $v_j$. Since the length of the period vector $\pervec_{j+1}$ is either the same as the period vector $\pervec_{j}$ or double the length of period vector $\pervec_{j}$, it suffices to verify that $x^{2^{\lfloor\frac{j}{2}\rfloor + 1}} = 0$ (assuming that we did the check for $\pervec_{j}$ already). Moreover, we also check that for even $j$ we have  $x^i_j = 1-x^{i''}_j$, where $i'' = i+ 2^{\lfloor\frac{j}{2}\rfloor}$.
\begin{itemize}
\item $\pervec_1 = (01)$;
\item $\pervec_2 = (01)(10) = (011)(0)$;
\item $\pervec_3 = (001)(0)$;
\item $\pervec_4 = (010)(0)(101)(1)=(010)^2(11)$;
\item $\pervec_5 = (001)^2(01)$;
\item $\pervec_6 = (010)^2(01)(101)^2(10) = (010)^2(011)^3(0)$;
\item $\pervec_7 = (001)^5(0)$;
\item $\pervec_8 = (010)^5(0)(101)^5(1)= (010)^6(110)^4(11)$;
\item $\pervec_9 = (001)^6(010)^4(01)$;
\item $\pervec_{10} = (010)^6(011)^4(01)(101)^6(100)^4(10)$ \\ $~~~~~~~~~ = (010)^6(011)^{11}(001)^4(0)$;
\item $\pervec_{11} = (001)^{17}(000)(100)^3(1)$;
\item $\pervec_{12} = (010)^{17}(010)(110)^3(1)(101)^{17}(101)(001)^3(0)$ 
\\ $~~~~~~~~~= (010)^{18}(110)^{21}(100)^3(10)$;
\end{itemize}
The reason we computed the first $12$ entries is that we will show that from here onward, the structure of the period vectors for $v_j$ will begin to follow a cyclic pattern. Namely, we will distinguish the remaining nodes by $(j\mod 4)$ and show that the structure of their period vectors is the same as the structure of period vectors for $v_9$, $v_{10}$, $v_{11}$, and $v_{12}$ respectively. 
More precisely, we show that for $j \ge 9$:

\begin{itemize}
\item If $j=4k$, $\pervec_j=(010)^\ell (110)^p(100)^q(10)$ for some $\ell, p, q \in \mathbb{N}$ such that $\ell + p + q = \frac{2^{2k+1}-2}{3}.$

\item If $j=4k+1$, $\pervec_j=(001)^p(010)^q(01)$ for some $p, q \in \mathbb{N}$ such that $p + q = \frac{2^{2k+1}-2}{3}.$

\item If $j=4k+2$, $\pervec_j=(010)^\ell (011)^p(001)^q(0)$ for some $\ell, p, q \in \mathbb{N}$ such that $\ell + p + q = \frac{2^{2k+2}-1}{3}.$

\item If $j=4k+3$, $\pervec_j=(001)^p(000)(100)^q(1)$ for some $p, q \in \mathbb{N}$ such that $p + q = \frac{2^{2k+2}-4}{3}.$

\end{itemize}

\ifshort
To complete the proof of the claim, it suffices to verify the period vectors have the stated structure and that each of the cases satisfies the stated conditions.
\fi
\iflong
While proving the structure of the period vector in each case, we will simultaneously show that the claim holds in each of the cases. Moreover, note that the structural claim about $\pervec_j$ holds for $j=12$.  
We distinguish the cases on the remainder of $j$ modulo four.

\paragraph{Case $j=4k$ for $k\ge 3$.} 
Let us assume that
$\pervec_j=(010)^\ell (110)^p(100)^q(10)$ for some $\ell, p, q \in \mathbb{N}$ such that $\ell + p + q = \frac{2^{2k+1}-2}{3}.$
We know that $f_{j+1}(b_1,b_2) = b_1\cdot (1-b_2)$. Moreover $x^0_{j+1} = 0$. Before we prove the structure for the case $j+1=4k+1$ more formally, let us write the two structures under each other in a suggestive way such that it is rather straightforward to verify the structure of the period vector $\pervec_{j+1}$ from $\pervec_{j+1}$, fact that $x^0_{j+1}=0$, and the function $f_{j+1}$.
\begin{align*}
\pervec_{4k} &\rightarrow (010)^\ell (110)^p(100)^q(10)\\ 
\pervec_{4k+1} &\rightarrow (001)^\ell (010)^p (010)^q(01)
\end{align*}
Now, we discuss this a bit more formally.
Let us first check how the state of $v_{j+1}$ changes, when the state of $v_j$ is in the first part where it changes as $(010)^\ell$ staring from $x^0_j=0$ and $x^0_{j+1}=0$. We get  $x^1_{j+1} = f_{j+1}(x^0_j, x^0_{j+1})= f_{j+1}(0,0) = 0$, $x^2_{j+1} = f_{j+1}(1,0) = 1$ and $x^3_{j+1} = f_{j+1}(0,1) = 0$. That is from the first $(010)$ for $v_j$ we get $(001)$ and in addition we get that the next triple starts again with $0$. It follows that $\pervec_{j+1}$ starts with $(001)^\ell$ and $x^{3\ell}_{j+1}=0$. Given this, let us see how $x^i_{j+1}$ changes on the
$(110)^p$. We have $f_{j+1}(1,0)=1$, $f_{j+1}(1,1)=0$, and $f_{j+1}(0,0)=0$. Hence we get $(010)$ on the first $(110)$ with addition that $x^{3\ell+3}_{j+1}=0$, hence we get additional $(010)^p$ to $\pervec_{j+1}$. For $(100)$, when the first state of $v_{j+1}$ is $0$, we get $f_{j+1}(1,0)= 1$, $f_{j+1}(0,1)= 0$, $f_{j+1}(0,0)= 0$ . Hence, we get $(010)$ for the first $(100)$ with the next state of $v_{j+1}$ being again $0$. Therefore, we get additional $(010)^q$. Finally for $(10)$ we get $f_{j+1}(1,0)=1$ which leads to final $(01)$ to the period vector $\pervec_{j+1}$ and $\pervec_{j+1}= (001)^\ell(010)^{p+q}(01)$. Finally, $f(0,1)=0$ ensures that $x^{2^{2k+1}}_{j+1}=x^0_{j+1}$ and since, by induction hypothesis, the period vector for $v_j$ repeats, so does the vector for $v_{j+1}$. 

\paragraph{Case $j=4k+1$ for $k\ge 3$.} Let us assume that
$\pervec_j=(001)^p(010)^q(01)$ for some $p, q \in \mathbb{N}$ such that $p + q = \frac{2^{2k+1}-2}{3}.$
We know that $f_{j+1}(b_1,b_2) = (b_1-b_2+1)\cdot (b_2-b_1+1)$. Moreover $x^0_{j+1} = 0$. Let us again write down how the period vector $\pervec_{j}$ influenced the period vector $\pervec_{j+1}$.
\begin{align*}
\pervec_{4k+1} &\rightarrow (001)^p(010)^q(01)\mid (001)^p(010)^q(01) \\ 
\pervec_{4k+2} &\rightarrow (010)^p(011)^q(01)\mid (101)^p(100)^q(10) \\
\pervec_{4k+2} &\rightarrow (010)^p(011)^q(01~~~~~1) (011)^p(001)^q(0)
\end{align*} 
Note that $f_{j+1}(1,1)=1$, hence $x^{2^{2k+1}}_{j+1}=1$ and we needed to repeat the period vector for $v_j$ twice to get the period vector for $v_{j+1}$. Moreover, we can see that for all $i\in \{0,\ldots, 2^{2k+1}-1\}$ it holds that $x^{i}_{j+1} = 1- x^{i+2^{2k+1}}_{j+1}$. However, $f_{j+1}(1,0)=0$ and $x^{2^{2k+2}}_{j+1}=x^0_{j+1}=0$ and it follows that the period vector for $v_{j+1}$ repeats each $2^{2k+2}$ steps. Checking that the transition outlined above is correct is rather straightforward and similar to the case $j=4k$. 

\paragraph{Case $j=4k+2$ for $k\ge 3$.} Assume that $\pervec_j=(010)^\ell (011)^p(001)^q(10)$ for some $\ell, p, q \in \mathbb{N}$ such that $\ell + p + q = \frac{2^{2k+2}-1}{3}.$ Moreover, $f_{j+1}(b_1,b_2) = b_1\cdot (1-b_2)$ and $x^0_{j+1} = 0$. We get the following transition from $\pervec_{j}$ to $\pervec_{j+1}$
\begin{align*}
\pervec_{4k+2} &\rightarrow (010)^\ell(011)^p(001)(001)^{q-1}(0) \\
\pervec_{4k+3} &\rightarrow (001)^\ell(001)^p(000)(100)^{q-1}(1) \\
\end{align*} 
Since $f_{j+1}(0,1)= 0$, we get $x^{2^{2k+2}}_{j+1}=x^0_{j+1}=0$ and it follows that the period vector for $v_{j+1}$ repeats each $2^{2k+2}$ steps. Checking that the transition outlined above is correct is rather straightforward and similar to the case $j=4k$. 

\paragraph{Case $j=4k+3$ for $k\ge 3$.} Assume that $\pervec_j=(001)^p(000)(100)^q(1)$ for some $p, q \in \mathbb{N}$ such that $p + q = \frac{2^{2k+2}-1}{3}.$ Moreover, $f_{j+1}(b_1,b_2) = (b_1-b_2+1)\cdot (b_2-b_1+1)$ and $x^0_{j+1} = 0$. We get the following transition from $\pervec_{j}$ to $\pervec_{j+1}$
\begin{align*}
\pervec_{4k+3} &\rightarrow (001)^p(000)(100)^q(1)\mid (001)^p(000)(100)^q(1) \\
\pervec_{4k+4} &\rightarrow (010)^p(010)(110)^q(1)\mid (101)^p(101)(001)^q(0) \\
\pervec_{4k+4} &\rightarrow (010)^p(010)(110)^q(1~~~~~~10) (110)^p(100)^q(10)
\end{align*} 
Again, we have that $f_{j+1}(1,1)=1$, hence $x^{2^{2k+2}}_{j+1}=1$ and we needed to repeat the period vector for $v_j$ twice to get the period vector for $v_{j+1}$. Moreover, we can see that for all $i\in \{0,\ldots, 2^{2k+2}-1\}$ it holds that $x^{i}_{j+1} = 1- x^{i+2^{2k+2}}_{j+1}$. Finally, $f_{j+1}(1,0)=0$ and $x^{2^{2k+2}}_{j+1}=x^0_{j+1}=0$ and it follows that the period vector for $v_{j+1}$ repeats each $2^{2k+3}$ steps.
\end{claimproof}
\fi
\ifshort
\end{claimproofsketch}
\fi

Observe that the above claim already implies that for all $i\in [n+1]$ and $p\in \mathbb{N}$, the state of the node $v_{2i-2}$ after $2^i\cdot p$ steps is $0$ (i.e., $x^{2^i\cdot p}_{2i-2}=0$). 

Now, given the above claim on the structure of the configuration, it is also rather straightforward to show that for all $i\in [n+1]$ and all tuples $t\in \dom^{i}$, there exists $q\in \{0,\ldots, 2^i-1\}$ such that, for every $p\in \mathbb{N}$, the configuration of $\XXX^{(n)}$\ after $2^i\cdot p + q$ successor steps restricted to the nodes $v_1, v_2, v_4, \ldots, v_{2i-4}, v_{2i-2}$ is equal to $t$. First note that for $i=1$, the node $v_1$ flips always between $0$ and $1$ and the statement holds. Moreover, for $i=2$, the nodes $v_1$ and $v_2$ have the following transitions through configurations $00\rightsquigarrow 11\rightsquigarrow 01\rightsquigarrow 10\rightsquigarrow 00$. Let us assume that for some $i\in [n]$ for every tuple $t\in \dom^{i}$, there exists $q\in \{0,\ldots, 2^i-1\}$ such that for every $p\in \mathbb{N}$, the configuration of $\XXX^{(n)}$\ after $2^i\cdot p + q$ successor steps restricted to the nodes $v_1, v_2, v_4, \ldots, v_{2i-4}, v_{2i-2}$ is equal to $t$.

Let $t\in \dom^{i+1}$, we will show that there exists $q\in \{0,\ldots, 2^{i+1}-1\}$ such that for every $p\in \mathbb{N}$, the configuration of $\XXX^{(n)}$\ after $2^{i+1}\cdot p + q$ successor steps restricted to the nodes $v_1, v_2, v_4, \ldots, v_{2i-2}, v_{2i}$ is equal to $t$.
Denote by $t'$ the restriction of $t$ to the first $i$ bits. By induction hypothesis, there exists $q'\in \{0,\ldots, 2^i-1\}$ such that for every $p'\in \mathbb{N}$, the configuration of $\XXX^{(n)}$\ after $2^i\cdot p' + q'$ successor steps restricted to the nodes $v_1, v_2, v_4, \ldots, v_{2i-4}, v_{2i-2}$ is equal to $t'$. By the above claim, $$x_{2i}^{2^i\cdot p' + q'} = 1- x_{2i}^{2^i\cdot p' + q' + 2^i}.$$ It follows that for every $p\in \mathbb{N}$, $t$ appears as the restriction of the configuration to $v_1, v_2, v_4, \ldots, v_{2i-2}, v_{2i}$ after $2^{i+1}\cdot p + q$ steps, where $q$ is either equal to $q'$ or to $q'+2^i$. 
Repeating the same argument for all $t\in \dom^{i+1}$ completes the proof.
\end{proof}

With the path-gadget ready, we can proceed to establishing our first hardness result. The idea of the reduction used in the proof loosely follows that of previous work~\cite{RosenkrantzMRS21}, but uses multiple copies of the path gadget from \Cref{thm:paths} (instead of a single counting gadget that has a tournament as the network) to avoid dense substructures.
Note that the usage of the structurally simpler, but behaviourally more complicated, path-gadget requires some additional changes to the construction. 
Furthermore, compared to the earlier result of Barrett et al.~\cite{BarrettHMRRS06}, we obtain much smaller bounds on the treewidth and maximum degree (on the other hand, our reduction does not provide a bound on the bandwidth of the graph).

\thmpspace*
\iflong
\begin{proof}
\fi
\ifshort
\begin{proof}[Proof Sketch]
\fi
	Both problems are known and easily seen to be in \PSPACE~\cite{ChistikovLPT20}.
	For showing hardness, we give a reduction starting from \textsc{QBF} which is known to be \PSPACE-hard~\cite{GareyJ79} and takes as input a formula \(Q_n x_n \dots Q_1 x_1 \varphi\) where all \(Q_i\) are either \(\forall\) or \(\exists\) and $\varphi$ is a \textsc{3SAT}-formula.
	
	From this we construct a \SDSys\ \(\SSS\) which converges to one specific configuration (the all-one configuration) if and only if the given formula is true.
	
	The overall idea is for (multiple copies\footnote{all copies will be independent, have no additional in-neighbours and start at the same configuration} of) the \SDSys \(\XXX^{(n)}\) (see \Cref{thm:paths}) to iterate over all possible truth assignments of the variables \(x_1\) to \(x_n\).
	While it would be possible to simply associate every significant node \(\XXX^{(n)}\) with a variable, we will add auxiliary \emph{variable nodes} as to go through the variable assignments in a well-structured way.
	Then \emph{clause nodes} aggregate the implied truth values of the clauses of \(\varphi\).
	These clause nodes can then be used as incoming neighbours for \emph{subformula nodes} to derive the truth value of formulas \(Q_i x_i \dots Q_1 x_1 \varphi\) respectively.
	Once the truth value of \(Q_n x_n \dots Q_1 x_1 \varphi\) is determined to be true we use a \emph{control node} to set the whole \SDSys to a configuration with all entries being \(1\).
	Otherwise the copies of \(\XXX^{(n)}\) continue to loop through their period.
	We call the nodes \(v_{1}\) and \(v_{2i}\) with \(i \in [n]\) of a copy of \(\XXX^{(n)}\) \emph{significant}.
	
	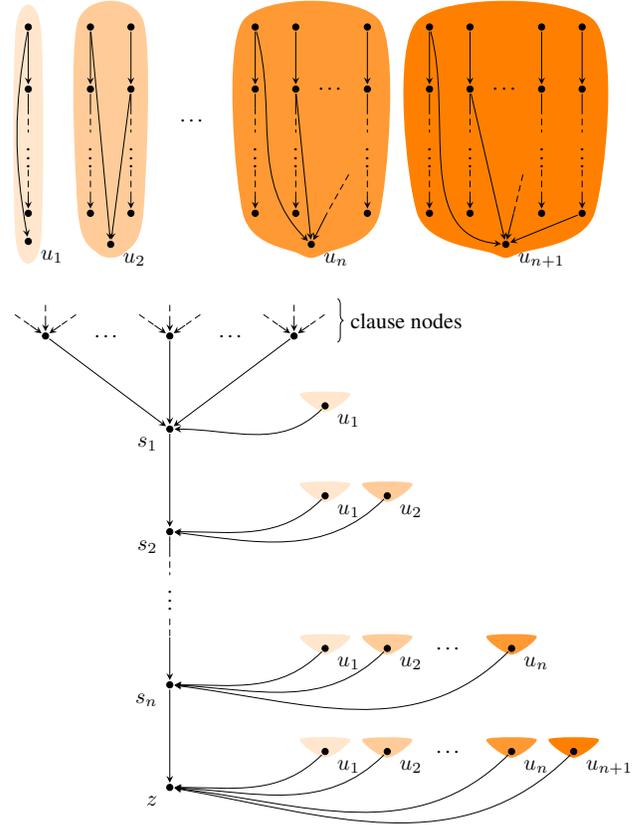
\begin{figure}[t]
		\resizebox{\columnwidth}{!}{		\begin{tikzpicture}
		\tikzstyle{vnodev}=[draw, shape=circle, inner sep=1pt, outer sep =0.75pt, fill=black]
						\node[vnodev] (v1) {};
			\node[vnodev, below of= v1] (v2) {};
			\node[below of=v2] (dots) {\(\vdots\)};
			\node[vnodev, below of=dots] (vn1) {};
			
			\draw[-stealth] (v1)--(v2);
			\draw[-] (v2)--($(dots)+(0,0.7)$);
			\draw[densely dashed] ($(dots)+(0,0.7)$) -- ($(dots)+(0,0.3)$);
			\draw[densely dashed] ($(vn1)+(0,0.7)$) -- ($(vn1)+(0,0.3)$);
			\draw[-stealth] ($(vn1)+(0,0.3)$) -- (vn1);
			
			\node[vnodev,label=-10:\(u_1\),below=0.3 of vn1] (u1) {};
			\draw[-stealth] (v1) to[out = -100, in = 100] (u1);
			\begin{scope}[on background layer]
			\fill [orange!20] plot [smooth cycle, tension=2] coordinates { ($(u1)-(0.2,0)$) ($(u1)-(-0.2,0)$) ($(v1) - (-0.2,-0)$) ($(v1) - (0.2,-0)$)};
			\end{scope}
			
						\node[vnodev] (v1') at ($(v1) + (1,0)$) {};
			\node[vnodev, below of= v1'] (v2') {};
			\node[below of=v2'] (dots') {\(\vdots\)};
			\node[vnodev, below of=dots'] (vn1') {};
			
			\node[vnodev,right=0.5 of v1'] (v1'2) {};
			\node[vnodev, below of= v1'2] (v2'2) {};
			\node[below of=v2'2] (dots'2) {\(\vdots\)};
			\node[vnodev, below of=dots'2] (vn1'2) {};
			
			\draw[-stealth] (v1')--(v2');
			\draw[-] (v2')--($(dots')+(0,0.7)$);
			\draw[densely dashed] ($(dots')+(0,0.7)$) -- ($(dots')+(0,0.3)$);
			\draw[densely dashed] ($(vn1')+(0,0.7)$) -- ($(vn1')+(0,0.3)$);
			\draw[-stealth] ($(vn1')+(0,0.3)$) -- (vn1');
			
			\draw[-stealth] (v1'2)--(v2'2);
			\draw[-] (v2'2)--($(dots'2)+(0,0.7)$);
			\draw[densely dashed] ($(dots'2)+(0,0.7)$) -- ($(dots'2)+(0,0.3)$);
			\draw[densely dashed] ($(vn1'2)+(0,0.7)$) -- ($(vn1'2)+(0,0.3)$);
			\draw[-stealth] ($(vn1'2)+(0,0.3)$) -- (vn1'2);
			
			\node[vnodev,label=-10:\(u_2\)] (u2) at ($(vn1')!0.5!(vn1'2) - (0, 0.5)$) {};
			\draw[-stealth] (v1')--(u2);
			\draw[-stealth] (v2'2)--(u2);
			\begin{scope}[on background layer]
			\fill [orange!40] plot [smooth cycle, tension=.5] coordinates { ($(u2)-(0.2,0.15)$) ($(u2)-(-0.2,0.15)$) ($(vn1'2) - (-0.2,0)$) ($(v1'2) - (-0.2,-0.1)$) ($(v1') - (0.2,-0.1)$) ($(vn1') - (0.2,0)$)};
			\end{scope}
			
						\node[vnodev] (v1'') at ($(v1'2) + (2,0)$) {};
			\node[vnodev, below of= v1''] (v2'') {};
			\node[below of=v2''] (dots'') {\(\vdots\)};
			\node[vnodev, below of=dots''] (vn1'') {};
			
			\node[vnodev,right=0.5 of v1''] (v1''2) {};
			\node[vnodev, below of= v1''2] (v2''2) {};
			\node[below of=v2''2] (dots''2) {\(\vdots\)};
			\node[vnodev, below of=dots''2] (vn1''2) {};
			
			\node[vnodev,right= of v1''2] (v1''n) {};
			\node[vnodev, below of= v1''n] (v2''n) {};
			\node[below of=v2''n] (dots''n) {\(\vdots\)};
			\node[vnodev, below of=dots''n] (vn1''n) {};
			
			\draw[-stealth] (v1'')--(v2'');
			\draw[-] (v2'')--($(dots'')+(0,0.7)$);
			\draw[densely dashed] ($(dots'')+(0,0.7)$) -- ($(dots'')+(0,0.3)$);
			\draw[densely dashed] ($(vn1'')+(0,0.7)$) -- ($(vn1'')+(0,0.3)$);
			\draw[-stealth] ($(vn1'')+(0,0.3)$) -- (vn1'');
			
			\draw[-stealth] (v1''2)--(v2''2);
			\draw[-] (v2''2)--($(dots''2)+(0,0.7)$);
			\draw[densely dashed] ($(dots''2)+(0,0.7)$) -- ($(dots''2)+(0,0.3)$);
			\draw[densely dashed] ($(vn1''2)+(0,0.7)$) -- ($(vn1''2)+(0,0.3)$);
			\draw[-stealth] ($(vn1''2)+(0,0.3)$) -- (vn1''2);
			
			\draw[-stealth] (v1''n)--(v2''n);
			\draw[-] (v2''n)--($(dots''n)+(0,0.7)$);
			\draw[densely dashed] ($(dots''n)+(0,0.7)$) -- ($(dots''n)+(0,0.3)$);
			\draw[densely dashed] ($(vn1''n)+(0,0.7)$) -- ($(vn1''n)+(0,0.3)$);
			\draw[-stealth] ($(vn1''n)+(0,0.3)$) -- (vn1''n);
			
			\node[vnodev,label=-10:\(u_{n}\)] (un) at ($(vn1'')!0.5!(vn1''n) - (0, 0.5)$) {};
			\draw[-stealth] (v1'') to [out=-75,in=135] (un);
			\draw[-stealth] (v2''2)--(un);
			\draw[densely dashed] ($(dots''n) - (0.3,0.375)$)--($(dots''n) - (0.65,1)$);
			\draw[-stealth] ($(dots''n) - (0.65,1)$)--(un);
			\begin{scope}[on background layer]
				\fill [orange!80] plot [smooth cycle, tension=.2] coordinates { ($(un)-(0.2,0.15)$) ($(un)-(-0.2,0.15)$) ($(vn1''n) - (-0.2,0)$) ($(v1''n) - (-0.2,-0.1)$) ($(v1'') - (0.2,-0.1)$) ($(vn1'') - (0.2,0)$)};
			\end{scope}
		
		\node at ($(v2''2)!0.5!(v2''n)$) {\dots};
		
				\node[vnodev] (v1''') at ($(v1''n) + (1,0)$) {};
		\node[vnodev, below of= v1'''] (v2''') {};
		\node[below of=v2'''] (dots''') {\(\vdots\)};
		\node[vnodev, below of=dots'''] (vn1''') {};
		
		\node[vnodev,right=0.5 of v1'''] (v1'''2) {};
		\node[vnodev, below of= v1'''2] (v2'''2) {};
		\node[below of=v2'''2] (dots'''2) {\(\vdots\)};
		\node[vnodev, below of=dots'''2] (vn1'''2) {};
		
		\node[vnodev,right= of v1'''2] (v1'''n) {};
		\node[vnodev, below of= v1'''n] (v2'''n) {};
		\node[below of=v2'''n] (dots'''n) {\(\vdots\)};
		\node[vnodev, below of=dots'''n] (vn1'''n) {};
		
		\node[vnodev,right=0.5 of v1'''n] (v1'''n1) {};
		\node[vnodev, below of= v1'''n1] (v2'''n1) {};
		\node[below of=v2'''n1] (dots'''n1) {\(\vdots\)};
		\node[vnodev, below of=dots'''n1] (vn1'''n1) {};
		
		\draw[-stealth] (v1''')--(v2''');
		\draw[-] (v2''')--($(dots''')+(0,0.7)$);
		\draw[densely dashed] ($(dots''')+(0,0.7)$) -- ($(dots''')+(0,0.3)$);
		\draw[densely dashed] ($(vn1''')+(0,0.7)$) -- ($(vn1''')+(0,0.3)$);
		\draw[-stealth] ($(vn1''')+(0,0.3)$) -- (vn1''');
		
		\draw[-stealth] (v1'''2)--(v2'''2);
		\draw[-] (v2'''2)--($(dots'''2)+(0,0.7)$);
		\draw[densely dashed] ($(dots'''2)+(0,0.7)$) -- ($(dots'''2)+(0,0.3)$);
		\draw[densely dashed] ($(vn1'''2)+(0,0.7)$) -- ($(vn1'''2)+(0,0.3)$);
		\draw[-stealth] ($(vn1'''2)+(0,0.3)$) -- (vn1'''2);
		
		\draw[-stealth] (v1'''n)--(v2'''n);
		\draw[-] (v2'''n)--($(dots'''n)+(0,0.7)$);
		\draw[densely dashed] ($(dots'''n)+(0,0.7)$) -- ($(dots'''n)+(0,0.3)$);
		\draw[densely dashed] ($(vn1'''n)+(0,0.7)$) -- ($(vn1'''n)+(0,0.3)$);
		\draw[-stealth] ($(vn1'''n)+(0,0.3)$) -- (vn1'''n);
		
		\draw[-stealth] (v1'''n1)--(v2'''n1);
		\draw[-] (v2'''n1)--($(dots'''n1)+(0,0.7)$);
		\draw[densely dashed] ($(dots'''n1)+(0,0.7)$) -- ($(dots'''n1)+(0,0.3)$);
		\draw[densely dashed] ($(vn1'''n1)+(0,0.7)$) -- ($(vn1'''n1)+(0,0.3)$);
		\draw[-stealth] ($(vn1'''n1)+(0,0.3)$) -- (vn1'''n1);
		
		\node[vnodev,label=-10:\(u_{n + 1}\)] (un1) at ($(vn1''')!0.5!(vn1'''n1) - (0, 0.5)$) {};
		\draw[-stealth] (v1''') to [out=-75,in=175] (un1);
		\draw[-stealth] (v2'''2)--(un1);
		\draw[densely dashed] ($(dots'''n) - (0.3,0.375)$)--($(dots'''n) - (0.45,1)$);
		\draw[-stealth] ($(dots'''n) - (0.45,1)$)--(un1);
		\draw[-stealth] (vn1'''n1)--(un1);
		\begin{scope}[on background layer]
			\fill [orange!100] plot [smooth cycle, tension=0.2] coordinates { ($(un1)-(0.2,0.15)$) ($(un1)-(-0.2,0.15)$) ($(vn1'''n1) - (-0.2,0)$) ($(v1'''n1) - (-0.2,-0.1)$) ($(v1''') - (0.2,-0.1)$) ($(vn1''') - (0.2,0)$)};
		\end{scope}
	
		\node at ($(v2'''2)!0.5!(v2'''n)$) {\dots};

			\node at ($(vn1'2)!0.5!(v1'')$) {\dots};
			
						\node[vnodev] (c2) at ($(u1)!0.5!(un) - (0, 1.5)$) {};
			\node[vnodev] (c1) at ($(c2) - (2, 0)$) {};
			\node[vnodev] (c3) at ($(c2) + (2, 0)$) {};
			
			\node at ($(c1)!0.5!(c2)$) {\dots};
			\node at ($(c2)!0.5!(c3)$) {\dots};
			
			\draw[-stealth] ($(c1)+(145:0.3)$)--(c1);
			\draw[densely dashed] ($(c1)+(145:0.6)$)--($(c1)+(145:0.3)$);
			\draw[-stealth] ($(c1)+(90:0.25)$)--(c1);
			\draw[densely dashed] ($(c1)+(90:0.5)$)--($(c1)+(90:0.25)$);
			\draw[-stealth] ($(c1)+(35:0.3)$)--(c1);
			\draw[densely dashed] ($(c1)+(35:0.6)$)--($(c1)+(35:0.3)$);
			
			\draw[-stealth] ($(c2)+(145:0.3)$)--(c2);
			\draw[densely dashed] ($(c2)+(145:0.6)$)--($(c2)+(145:0.3)$);
			\draw[-stealth] ($(c2)+(90:0.25)$)--(c2);
			\draw[densely dashed] ($(c2)+(90:0.5)$)--($(c2)+(90:0.25)$);
			\draw[-stealth] ($(c2)+(35:0.3)$)--(c2);
			\draw[densely dashed] ($(c2)+(35:0.6)$)--($(c2)+(35:0.3)$);
			
			\draw[-stealth] ($(c3)+(145:0.3)$)--(c3);
			\draw[densely dashed] ($(c3)+(145:0.6)$)--($(c3)+(145:0.3)$);
			\draw[-stealth] ($(c3)+(90:0.25)$)--(c3);
			\draw[densely dashed] ($(c3)+(90:0.5)$)--($(c3)+(90:0.25)$);
			\draw[-stealth] ($(c3)+(35:0.3)$)--(c3);
			\draw[densely dashed] ($(c3)+(35:0.6)$)--($(c3)+(35:0.3)$);
			
			\draw [decorate,
			decoration = {brace}] ($(c3) + (0.7,0.6)$) -- node[midway,right] {\ clause nodes} ($(c3) + (0.7,-0.1)$);
			
						\node[vnodev,label=190:\(s_1\)](s1) at ($(c2) - (0,1.5)$) {};
			\node[vnodev,below=1.5cm of s1,label=190:\(s_2\)](s2) {};
			\node[below of=s2] (vdots) {\(\vdots\)};
			\node[vnodev,below=1cm of vdots,label=190:\(s_n\)](sn) {};
			
			\draw[-stealth] (c1)--(s1);
			\draw[-stealth] (c2)--(s1);
			\draw[-stealth] (c3)--(s1);
			
			\draw[-stealth] (s1)--(s2);
			\draw (s2)--($(vdots)+(0,0.7)$);
			\draw[densely dashed] ($(vdots)+(0,0.7)$)--($(vdots)+(0,0.3)$);
			\draw[densely dashed] ($(vdots)-(0,0.4)$)--($(vdots)-(0,0.7)$);
			\draw[-stealth] ($(vdots)-(0,0.7)$)--(sn);
			
			\node[vnodev,below=1.5cm of sn,label=190:\(z\)](z) {};
			\draw[-stealth] (sn)--(z);
			
			\node[vnodev,label=-10:\(u_1\)](u1s1) at ($($(c2)!0.75!(s1)$) + (2.5,0)$) {};
			\begin{scope}[on background layer]
			\fill [orange!20] plot [smooth cycle, tension=.5] coordinates { ($(u1s1)-(0,0.1)$) ($(u1s1) - (0.4,-0.2)$) ($(u1s1) - (-0.4,-0.2)$)};
			\end{scope}
			\node[vnodev,label=-10:\(u_1\)](u1s2) at ($($(s1)!0.65!(s2)$) + (2.5,0)$) {};
			\begin{scope}[on background layer]
			\fill [orange!20] plot [smooth cycle, tension=.5] coordinates { ($(u1s2)-(0,0.1)$) ($(u1s2) - (0.4,-0.2)$) ($(u1s2) - (-0.4,-0.2)$)};
			\end{scope}
			\node[vnodev,label=-10:\(u_2\),right of=u1s2](u2s2) {};
			\begin{scope}[on background layer]
			\fill [orange!40] plot [smooth cycle, tension=.5] coordinates { ($(u2s2)-(0,0.1)$) ($(u2s2) - (0.4,-0.2)$) ($(u2s2) - (-0.4,-0.2)$)};
			\end{scope}
			\node[vnodev,label=-10:\(u_1\)](u1sn) at ($($(vdots)!0.6!(sn)$) + (2.5,0)$) {};
			\begin{scope}[on background layer]
			\fill [orange!20] plot [smooth cycle, tension=.5] coordinates { ($(u1sn)-(0,0.1)$) ($(u1sn) - (0.4,-0.2)$) ($(u1sn) - (-0.4,-0.2)$)};
			\end{scope}
			\node[vnodev,label=-10:\(u_2\),right of=u1sn](u2sn) {};
			\begin{scope}[on background layer]
			\fill [orange!40] plot [smooth cycle, tension=.5] coordinates { ($(u2sn)-(0,0.1)$) ($(u2sn) - (0.4,-0.2)$) ($(u2sn) - (-0.4,-0.2)$)};
			\end{scope}
			\node[right of=u2sn] (dotssn) {\dots};
			\node[vnodev,label=-10:\(u_n\),right of=dotssn](unsn) {};
			\begin{scope}[on background layer]
			\fill [orange!80] plot [smooth cycle, tension=.5] coordinates { ($(unsn)-(0,0.1)$) ($(unsn) - (0.4,-0.2)$) ($(unsn) - (-0.4,-0.2)$)};
			\end{scope}
			\node[vnodev,label=-10:\(u_1\)](u1z) at ($($(sn)!0.65!(z)$) + (2.5,0)$) {};
			\begin{scope}[on background layer]
			\fill [orange!20] plot [smooth cycle, tension=.5] coordinates { ($(u1z)-(0,0.1)$) ($(u1z) - (0.4,-0.2)$) ($(u1z) - (-0.4,-0.2)$)};
			\end{scope}
			\node[vnodev,label=-10:\(u_2\),right of=u1z](u2z) {};
			\begin{scope}[on background layer]
			\fill [orange!40] plot [smooth cycle, tension=4] coordinates { ($(u2z)-(0,0.1)$) ($(u2z) - (0.4,-0.2)$) ($(u2z) - (-0.4,-0.2)$)};
			\end{scope}
			\node[right of=u2z] (dotsz) {\dots};
			\node[vnodev,label=-10:\(u_n\),right of=dotsz](unz) {};
			\begin{scope}[on background layer]
			\fill [orange!80] plot [smooth cycle, tension=.5] coordinates { ($(unz)-(0,0.1)$) ($(unz) - (0.4,-0.2)$) ($(unz) - (-0.4,-0.2)$)};
			\end{scope}
			\node[vnodev,label=-10:\(u_{n+1}\),right of=unz](un1z) {};
			\begin{scope}[on background layer]
			\fill [orange!100] plot [smooth cycle, tension=8] coordinates { ($(un1z)-(0,0.1)$) ($(un1z) - (0.4,-0.2)$) ($(un1z) - (-0.4,-0.2)$)};
			\end{scope}
			
			\draw[-stealth] (u1s1) to[out=225,in=0] (s1);
			\draw[-stealth] (u1s2) to[out=225,in=0] (s2);
			\draw[-stealth] (u1sn) to[out=225,in=0] (sn);
			\draw[-stealth] (u1z) to[out=225,in=0] (z);
			\draw[-stealth] (u2s2) to[out=225,in=-7] (s2);
			\draw[-stealth] (u2sn) to[out=225,in=-3] (sn);
			\draw[-stealth] (u2z) to[out=225,in=-3] (z);
			\draw[-stealth] (unsn) to[out=225,in=-7] (sn);
			\draw[-stealth] (unz) to[out=225,in=-7] (z);
			\draw[-stealth] (un1z) to[out=225,in=-10] (z);
		\end{tikzpicture}}
		\caption{Overview of the network constructed in the proof of \Cref{thm: pspace}.
		To not clutter the figure, the arcs from \(z\) to all other nodes omitted, and subgraphs introduced for variable nodes are illustrated in detail at the top and only indicated by orange (darker shade \(\simeq\) higher index) in all other places.}\label{fig:pspace-reduction}
	\end{figure}
	
	To not obfuscate the core of the construction we initially allow an arbitrary node degree and give a description on how to reduce it to a constant at the end of the proof.
	The overall structure of the underlying graph which we construct is depicted in \Cref{fig:pspace-reduction}.
	The starting configuration will have an entry \(0\) for all nodes except the ones in the copies of \(\XXX^{(n)}\) which will instead start at configuration \(x^1\) from \Cref{thm:paths}.
	We describe our construction in a hierarchical manner, going from the most fundamental to the most complex (in the sense that they rely on the earlier ones) types of nodes.
	
	For each variable \(x_i\), a variable node \(u_i\) has an incoming arc from the \(j\)-th significant node of the \(j\)-th of \(i\) private copies of \(\XXX^{(n)}\) for all \(j \in [i]\).
	In addition we add an auxiliary node \(u_{n + 1}\) which will not actually represent a variable of \(\varphi\) itself but otherwise be treated and hence also be referred to as a variable node for convenience.
	\(u_{n + 1}\) has incoming arcs from the \(j\)-th significant node of the \(j\)-th of \(n+1\) private copies of \(\XXX^{(n)}\) for all \(j \in [n + 1]\).
	It will become apparent later that it is useful to introduce \(u_{n + 1}\) and that \(X^ {(n)}\) cycles only after \(2^{n + 1}\) time steps rather than just \(2^n\) for technical reasons.
	We define \(f_{u_i}\) to change the configuration at \(u_i\) if and only if the configuration of all of its in-neighbours is \(0\) and to leave it unchanged otherwise.
	To give some intuition for the behaviour of the variable nodes with this construction it is useful to refer to the \emph{position} of a tuple \(t \in \dom^{\{v_1\} \cup \{v_{2i} \mid i \in [n]\}}\) as the number of time steps required to reach a configuration of \(\XXX^{(n)}\) that is equal to \(t\) on the significant nodes starting from the configuration \(x^1\). 
	\Cref{thm:paths} upper-bounds the position of a tuple by \(2^{n + 1}\).

\iflong
	\begin{claim}
	\fi
	\ifshort
	\begin{claim}
	\fi
		\label{claim:nice-counter}
		At any time step, taken together the variable nodes \(u_{n + 1} \dots u_1\) encode the position of the configuration of the significant nodes of the path gadgets in binary (that are the same on different copies) in the previous time step. 	\end{claim}
	\iflong
	\begin{claimproof}
		By \Cref{thm:paths}	and the choice of our starting configuration, the configuration at the \(i\) first significant nodes is all-zero exactly in multiples of \(2^i\) time steps.
		Because of the definition of the local functions this means that the configuration at \(u_i\) flips at each time step after a multiple of \(2^i\).
		Considering all variable nodes together, this is exactly how a cyclic binary counter from \(0\) to \(2^{n + 1}\) works which is the same as a counter for the position of a configuration on the significant nodes of \(\XXX^{(n)}\).
	\end{claimproof}
	\fi

	This means that the variable nodes just like the significant nodes of \(\XXX^{(n)}\) have the property that on them any shared configuration is reachable, and in addition these shared configurations are reached in the order of the binary numbers they represent.
	Clearly this construction for each \(u_i\) leads to an underlying tree and hence has treewidth \(1\).
		
	For each clause \(C = x_{j_1} \lor x_{j_2} \lor x_{j_3}\) of \(\varphi\), a clause node \(c\) has incoming arcs from private variable nodes \(u_{j_1}\), \(u_{j_2}\) and \(u_{j_3}\) and its local function is \(1\) if and only if the configuration of at least one of these in-neighbours is \(1\).
	This means we will have a copy of each variable node with its own copies of \(\XXX^{(n)}\) for each occurrence of that variable in \(\varphi\).
	The underlying graph formed by this construction is still a tree.
	
	For each \(i \in [n]\) we introduce a subformula node \(s_i\).
	\(s_1\) has incoming arcs from all clause nodes and a private variable node \(u_1\).
	For \(i \in [n] \setminus \{1\}\), \(s_i\) has incoming arcs from \(s_{i - 1}\) and private variable nodes \(u_1, \dotsc, u_i\).
	As before, this construction still results in a tree.
		
	Now we turn to the definition of the local functions for subformula nodes.
	We define \(f_{s_1}\) as follows:
	If the entry from \(u_1\) is \(1\), then \(f_{s_1}\) maps to \(1\) if and only if all entries from the clause nodes are \(1\).
	Otherwise if \(Q_1 = \exists\), \(f_{s_1}\) maps to \(1\) if and only if the entry for \(s_1\) \emph{or} all entries from the clause nodes are \(1\).
	Otherwise the entry from \(u_1\) is \(0\) and \(Q_1 = \forall\), and we let \(f_{s_1}\) map to \(1\) if and only if the entry for \(s_1\) \emph{and} all entries from the clause nodes are \(1\).	
	For \(i \in [n] \setminus \{1\}\) we define \(f_{s_i}\) as follows:
	If the incoming variable nodes in the order \(u_i \dots u_1\) are the binary encoding of \(2^{i - 1} + i - 1\) then \(f_{s_i}\) maps to the entry for \(s_{i - 1}\).
	Otherwise if the incoming variable nodes in the order \(u_i \dots u_1\) are the binary encoding of \(i - 1\) and \(Q_i = \exists\) then \(f_{s_i}\) maps to \(1\) if and only if the entry for \(s_i\) \emph{or} the entry for \(s_{i - 1}\) is \(1\).
	Otherwise if the incoming variable nodes in the order \(u_i \dots u_1\) are the binary encoding of \(i - 1\) and \(Q_i = \forall\) then \(f_{s_i}\) maps to \(1\) if and only if the entry for \(s_i\) \emph{and} the entry for \(s_{i - 1}\) is \(1\).
	Otherwise the incoming variable nodes in the order \(u_i \dots u_1\) do not encode \(2^{i - 1} + i - 1\) or \(i - 1\) and we let \(f_{s_i}\) map to the entry for \(s_i\).
	
	\ifshort
	It now remains to show (via a non-trivial inductive argument) that at time step $2^{n}+n$, the configuration at $s_n$ contains $1$ if and only if the \QBF\ input is true.
	\fi	
	\iflong
	With these local functions we can show the following.
	\begin{claim}[cf.\ {\cite[Claim~1]{RosenkrantzMRS21}}]
		Let \(i \in [n]\) and \(t\) be a tuple over \(\dom\) with an entry for each of \(\{x_{i + 1}, \dotsc, x_n\}\) which in the order \(t_{x_n} \dots t_{x_{i + 1}}\) is a binary encoding of the number \(k\).
		Then the configuration of \(s_i\) when \(u_n \dots u_1\) is the binary encoding of \((k + 1)2^{i} + i\) is equal to the truth value (i.e.\ \(1\) if \texttt{true} and \(0\) if \texttt{false}) of \(Q_i x_i \dots Q_1 x_1 \varphi(t)\) where \(\varphi(t)\) arises from \(\varphi\) by replacing each occurrence of \(x_j\) with \(j \in [n] \setminus [i]\) by its entry in \(t\).
	\end{claim}
	\begin{claimproof}
		This claim is already essentially shown as Claim~1 in \cite{RosenkrantzMRS21} but we reprove it in our setting and notation for self-containedness.

		The proof works by induction on \(i\).
		For the base case \(i = 1\) we want to show the correct configuration at \(s_1\) and when \(u_n \dots u_1\) encodes \((k + 1)2 + 1\).
		To determine the configuration of \(s_1\) when \(u_n \dots u_1\) encodes \((k + 1)2 + 1\) by \Cref{claim:nice-counter} we need to consider the application of the local function \(f_{s_1}\) when \(u_n \dots u_1\) encodes \((k + 1)2\).
		By the definition of \(f_{s_1}\) this means that then the configuration at \(s_1\) is derived as the as the disjunction if \(Q_1 = \exists\) and conjunction if \(Q_1 = \forall\) of the truth value of \(\varphi(a)\), where \(a\) associates \(x_n \dots x_1\) to the bits of the binary encoding of \((k + 1)2 - 1 = 2k + 1\), and the configuration at \(s_1\) when \(u_n \dots u_1\) encodes \((k + 1)2\).
		To determine the latter we have to consider the the application of the local function \(f_{s_1}\) when \(u_n \dots u_1\) encodes \((k + 1)2 - 1 = 2k + 1\).
		By the definition of \(f_{s_1}\) this means the configuration at \(s_1\) when \(u_n \dots u_1\) encodes \((k + 1)2\) is equal to the truth value of \(\varphi(b)\), where \(b\) associates \(x_n \dots x_1\) to the bits of the binary encoding of \(2k + 1 - 1 = 2k\).
		Observe that \(a\) and \(b\) correspond to \(t^{+1}\) which adds to \(t\) an entry with value \(1\) for \(x_1\) and \(t^{+0}\) which adds to \(t\) an entry with value \(0\) for \(x_1\) respectively.
		The truth value of \(Q_1 x_1 \varphi(t)\) is the disjunction of \(\varphi(t^{+0})\) and \(\varphi(t^{+1})\) if \(Q_1 = \exists\), and the conjunction of \(\varphi(t^{+0})\) and \(\varphi(t^{+1})\) if \(Q_1 = \forall\).
		This proves the base case.
		
		We now turn to the induction step.
		Similarly as before, the truth value of \(Q_{i + 1} x_{i + 1} \dots Q_1 x_1 \varphi(t)\) is the disjunction of \(Q_i x_i \dots Q_1 x_1 \varphi(t^{+0})\) and \(Q_i x_i \dots Q_1 x_1 \varphi(t^{+1})\) if \(Q_1 = \exists\), and the conjunction of \(Q_i x_i \dots Q_1 x_1 \varphi(t^{+0})\) and \(Q_i x_i \dots Q_1 x_1 \varphi(t^{+1})\) if \(Q_1 = \forall\) where \(t^{+0}\) adds an entry for \(x_{i + 1}\) with value \(0\) and \(t^{+1}\) adds an entry for \(x_{i + 1}\) with value \(1\) to \(t\).
		By induction hypothesis these are the configurations at \(s_i\) when \(u_n \dots u_1\) encodes \((2k + 1)2^i + i = 2^{i + 1}k + 2^i + i\) and \((2k + 2)2^i + i = 2^{i + 1}(k + 1) + i\) respectively.
		When \(u_n \dots u_1\) encodes \((2k + 1)2^i + i = 2^{i + 1}k + 2^i + i\) then the in-neighbours of \(s_{i + 1}\) together encode \(2^i + i = 2^{i + 1 - 1} + i + 1 - 1\).
		By the definition of \(f_{s_{i + 1}}\) then \(s_{i + 1}\) takes the configuration at \(s_i\) which by the argument above is the truth value of \(Q_i x_i \dots Q_1 x_1 \varphi(t^{+0})\).
		This configuration is retained at \(s_{i + 1}\) until \(u_n \dots u_1\) encodes \(2^{i + 1}(k + 1) + i\) because for step between \(u_n \dots u_1\) encoding \(2^{i + 1}k + 2^i + i\) and \(2^{i + 1}(k + 1) + i\), does \(x_i \dots x_1\) encode \(i\) or \(2^i + i\).
		When \(u_n \dots u_1\) encodes \(2^{i + 1}(k + 1) + i\) then the in-neighbours of \(s_{i + 1}\) together encode \(i = i + 1 - i\).
		By the definition of \(f_{s_{i + 1}}\) then \(s_{i + 1}\) takes the disjunction if \(Q_{i + 1} = \exists\) between the configuration at \(s_i\) which by the argument above is the truth value of \(Q_i x_i \dots Q_1 x_1 \varphi(t^{+1})\) and the configuration currently at \(s_{i + 1}\) which by the argument above is the truth value of \(Q_i x_i \dots Q_1 x_1 \varphi(t^{+0})\), and the conjunction of the same values if \(Q_{i + 1} = \forall\).
		This is exactly the desired behaviour described at the beginning of the induction step.
	\end{claimproof}
	
	Taking this, the \textsc{QBF} input formula is true if and only if the entry of the configuration of the \SDSys \(\SSS\) whose construction we described up to now in this proof when \(u_n \dots u_1\) are the binary encoding of	\(2^n + n\) (\(= (0 + 1)2^{n} + n\) for \(t\) being the empty tuple in the previous claim) is \(1\).

	\fi	
	In this case and only in this case we want the \SDSys to converge, and in particular interrupt the looping through different configurations of the copies of \(\XXX^{(n)}\).
	For this purpose we introduce a control node \(z\) which has an incoming arc from every node of a private copy of every variable node \(u_{n+1} \dots u_1\) and and an outgoing arc to all other nodes.
	We let \(f_z\) map to \(1\) if and only if the entries for the variable nodes encode \(2^n + n\) in binary (this is where the auxiliary node \(u_{n + 1}\) becomes important as this number cannot be encoded by the even-index nodes of only \(\XXX^{(n)}\)) and the entry for \(r_n\) is \(1\) or if the entry for \(f_z\) is already \(1\).
	We add an entry for \(f_z\) to all previously defined local functions and condition the previously defined behaviour on this entry being \(0\).
	As soon as this entry is \(1\) the local functions of all nodes should also map to \(1\).
	In this way the all-one configuration is reachable from the all-zero configuration if and only if the \textsc{QBF} input formula is true and this is the only way for the \SDSys to converge.
	
	Note that the network constructed in this manner consists of a directed tree plus the additional vertex $z$.

	\paragraph{Modification for constant in-degree.}
	To obtain a network with constant maximum in-degree, note that by \emph{currying}~\cite{Curry80} one can write every local function that takes polynomially many arguments as a polynomial-length series of functions that each take \(2\) arguments.
	In this way we can essentially replace every node with an in-degree of more than two by a tree of polynomial depth in which each node has in-degree \(2\).
	\iflong
	For simplicity, we keep the control node as in-neighbour for every node.
	Formally, we replace any node \(a\) with \(\delta^-(a) = \{b_1, \dotsc, b_\ell, z\}\) with \(\ell >2\) as follows.
	By currying, we find \(f_1, \dotsc, f_{\ell - 1}\) such that \(f_1, \dotsc, f_{\ell - 2}\) take three arguments each, \(f_{\ell - 1}\) takes four arguments and for any tuple \(t \in \dom^{\{b_1, \dotsc, b_\ell,b,z\}}\), \(f_a(t_{b_1}, \dotsc, t_{b_\ell}, t_b, t_z) = f_{\ell - 1}(f_{\ell - 2}(\dots f_1(t_{b_1}, t_{b_2}, t_z), t_{b_{\ell - 1}}, t_z), t_{b_{\ell}}, t_b, t_z)\).
	In place of \(a\) we introduce \(a_1, \dotsc, a_{\ell - 1}\) and let the in-neighbours of \(a_1\) be \(b_1\) and \(b_2\), and the in-neighbours of \(a_{i + 1}\) for \(i \in [\ell - 2] \setminus \{1\}\) be \(a_i\) and \(b_{i + 2}\).
	We let the out-neighbours of \(a_{\ell - 1}\) be the former out-neighbours of \(a\) and define \(f_{a_i}\) as \(f_i\).
	By the choice of \(f_{a_i}\) for \(i \in [\ell - 1]\), assuming synchronous computation of the configurations at all of \(a_1 \dotsc, a_{\ell - 1}\) from \(b_1 \dotsc, b_\ell\), \(a\) and \(a_{\ell - 1}\) would be functionally equivalent.
	
	However this assumption is not one we can simply make, as the configuration of \(b_1\) through \(b_\ell\) at one time step require differently long to reach \(a_{\ell - 1}\).
	To avoid asynchronous indirect arrival of the inputs from \(b_1, \dotsc, b_\ell\) we replace the arc that connects \(b_{i + 2}\) for \(i \in [\ell - 2]\) to the tree replacing \(a\) by a directed path with \(i - 1\) nodes.
	The first node on each of these paths behaves like \(a_i\) from the construction above and for the rest, let their local function copy the configuration of their respective predecessor at the previous time step.
	In this way, at time step \(t + \ell - 1\), \(a_{\ell - 1}\) is functionally equivalent to \(a\) from our old construction at time step \(t + 1\).
	This means the output at all nodes which stand for the nodes with in-degree \(\ell > 3\) from our original construction is delayed by \(\ell - 2\) time steps.
	Because \(\ell = n\) for all variable nodes, the configuration at every variable node is uniformly delayed by \(n - 2\) compared to our original construction.
	Similarly all clause nodes are delayed by \(1\) in addition to the cumulative delay of its non-\(z\) ancestors, \(s_1\) is delayed by one less than the number of clauses in addition to the cumulative delay of its non-\(z\) ancestors, and each \(s_i\) with \(i \in [n] \setminus \{1\}\) is delayed by \(i - 1\) in addition to the cumulative delay of its non-\(z\) ancestors.
	As all variable nodes remain synchronised, this delay does not impact the proof of correctness.	
	This replacement does not change the fact that the constructed network consists of a directed tree plus a single node.	
	\fi
		\end{proof}

\section{An Algorithm Using Treedepth}
In this section, we exploit the treedepth decomposition of the network to transform the instances of \reach\ and $\conv$ into equivalent ones of bounded size. Our algorithms then proceed by simulating all possible configurations of the resulting smaller network. 

Let $\mathcal{F}$ be the treedepth decomposition of $G$; without loss of generality, we may assume that $\mathcal{F}$ is a tree. We denote the subtree of $\mathcal{F}$ rooted in $v$ by $\mathcal{F}_v$. Let $G_v^*$ be the subgraph of $G$ induced by the nodes associated to $\mathcal{F}_v$ and the neighbourhood of $\mathcal F_v$ in $G$. 
Our aim is to iteratively compress nodes of $G$ from the leaves to the root to obtain a graph $G'$ with number of nodes bounded by some function of $\td(G)$.
\iflong
\begin{lemma}
\fi
\ifshort
\begin{lemma}
\fi
\label{lem: one_per_type}
Let $\III$ be an instance of $\conv$ or $\reach$ and let $\mathcal F$ be the treedepth decomposition of $G$ of height $L$. Assume that for some $v\in V(G)$, all the subtrees rooted in children of $v$ contain at most $m$ nodes. Then $\III$ can be modified in polynomial time to an equivalent instance $\III'$ with network $G'$
which has a treedepth decomposition $\mathcal F'$ of height at most $L$
such that $\mathcal F'\setminus \mathcal F'_v \subseteq \mathcal F \setminus \mathcal F_v$ and $|V(\mathcal F'_v)|\leq m^{m-1}\cdot(4^{L+m}\cdot |\dom|^{L+m+1})^m+1$.
\end{lemma} 
\iflong
\begin{proof}
\fi
\ifshort
\begin{proof}[Proof Sketch]
\fi
We say that two children $u$ and $w$ of $v$ have the same type (denoted by $u \equiv w$) if there exists an isomorphism $\phi: G_u^*\to G_w^*$ that is the identity on the neighbourhood of  \(\mathcal{F}_u\)  such that for every $z\in V(\mathcal F_u)$:
\begin{itemize}
 \item the initial states of $z$ and $\phi(z)$ are the same.
 \item $f_{\phi(z)}$ acts on $\dom^{\delta^-(\phi(z))\cup \{\phi(z)\}}$ as $f_z$ acts on $\dom^{\delta^-(z)\cup \{z\}}$  (the orderings of the neighbours of $z$ and their images agree).
\end{itemize}
In this case, the states of $z$ and $\phi(z)$ coinside at each time step.
Since the composition of two such isomorphisms results in an isomorphism with the same properties, $\equiv$ is an equivalence relation on the set of children of $v$ \ifshort with boundedly many equivalence classes\fi. 

\ifshort
Let $u$ and $w$ be the children of $v$ of the same type. Intuitively, since the configurations of the rooted subtrees $\mathcal F_u$ and $\mathcal F_v$ coincide at each time step, adapting the functions of the ancestors allows us to construct a new network which only preserves one of them.
\fi
\iflong
Let us upper-bound the number of its equivalence classes. For convenience, we will represent the isomorphisms between the rooted subtrees of children of $v$ by labelling their nodes. Namely, if $u_1,\dots u_s$ have the same type, let us fix some isomorphisms $\phi_i: G_{u_i}^*\to G_{u_{i+1}}^*$, $i\in [s-1]$, witnessing this. We label the nodes of $G_{u_1}$ arbitrarily and then for each $i\in [s-1]$ label the nodes of $G_{u_{i+1}}$ to make the labels preserved under $\phi_i$. By the Cayley's Formula (see, e.g., \cite{Shukla18}), the number of labelled trees on $m$ nodes is $m^{m-2}$. For every node $z$ of such a tree $Z$, there can be $|\dom|$ possible initial configurations. As $G \subseteq \lambda(\mathcal F)$, all the neighbours of $z$ are either its ancestors or belong to $Z$. Since $\mathcal F$ has the height of $L$, there are $2^{L+m-1}$ possibilities for the sets of in- and out-neighbours of $z$. This results in at most $2^{L+m}\cdot|\dom|^{L+m}$ options for the local function $f_z$. In total, the number of types of children of $v$ is at most $m^{m-2}\cdot(4^{L+m}\cdot |\dom|^{L+m+1})^m$.

Let $u$ and $w$ be the children of $v$ of the same type. Intuitively, since the configurations of the rooted subtrees $\mathcal F_u$ and $\mathcal F_v$ coincide at each time step, it is sufficient to preserve only one of them. Formally, let $\bar v_1, \dots, \bar v_q$ be the equivalence classes of $\equiv$ where $v_i\in \bar v_i$. Consider the graph $G'=G\setminus \bigcup_{i=1}^q \bigcup_{u_i\ne v_i} V(\mathcal F_{u_i})$ with the treedepth decomposition $\mathcal F'=\mathcal F[V(G')]$. We define the local function $f'_v$ for each node $v$ of $G'$ by shrinking its local function $f_v$ as follows. We say that the ordered tuple $\tau \in \dom^{\delta^-_{G}(v)\cup \{v\}}$ is an \emph{extension} of $\tau' \in \dom^{\delta^-_{G'}(v) \cup \{v\}}$ (or, equivalently, that $\tau'$ is a \emph{restriction} of $\tau$) if $\tau'$ can be obtained from $\tau$ by deleting the entries corresponding to the nodes that are not present in $G'$, while preserving the order. If in addition the entries of $\tau$ coincide for the nodes that belong to subtrees of children of $v$ with the same type and have the same label, we say that $\tau$ is a \emph{true extension} of $\tau'$.
Observe that for any $\tau' \in \dom^{\delta^-_{G'}(v)\cup \{v\}}$ there exists unique true extension $\tau$, we set $f_v'(\tau')=f_v(\tau)$.
Finally, we define the initial (and final, in case of $\reach$) configurations $x'$ ($y'$) as the restrictions of $x$ (and $y$ respectively) to the node set of $G'$.
In case of $\reach$, if $y$ is not a true extension of $y'$, we discard the constructed instance and return a trivial NO-instance. 
In the reduced instance, $v$ has at most $m^{m-2}\cdot(4^{L+m}\cdot |\dom|^{L+m+1})^m$ children. Since the rooted subtree of any child of $v$ has size of at most $m$, we can upper-bound $|V(\mathcal F'_v)|$ by $m^{m-1}\cdot(4^{L+m}\cdot |\dom|^{L+m+1})^m+1.$  
\fi
\end{proof}
\ifshort
Lemma \ref{lem: one_per_type} enables us to iteratively compress instances of $\conv$ or $\reach$.
\fi
\iflong
Lemma \ref{lem: one_per_type} allows us to iteratively compress instances of $\conv$ or $\reach$:

\begin{lemma} 
There exists a computable function $g:\Nat\times \Nat \to \Nat$ such that for any $L \in \Nat$ and $l\in[L]$, any instance $\III$ of $\conv$ or $\reach$ with $\td(G)=L$ can be in transformed in polynomial time into an equivalent  instance $\III'$ such that $\td(G')\le \td(G)$ and there is a treedepth decomposition of \(G'\) of height \(\td(G')\) in which any subtree rooted in a node of height $L-l+1$ has at most $g(L,l)$ nodes. 
\end{lemma}
\begin{proof}
For the leaves ($l=1$), we can simply set $g(L,1)=1$. Assume that $g(L,i)$ is defined for every $i\in [l]$. Let $\III$ be the instance of $\conv$ or $\reach$ and let $\mathcal F$ be the treedepth decomposition of $G$ of depth $L$ such that any subtree rooted in a node of height $L-l+1$ has at most $m=g(L,l)$ nodes. We consequently apply Lemma \ref{lem: one_per_type} to each node of level $L-l$. In a resulting graph, a subtree rooted in any node of level $L-l$ has size at most $m^{m-1}\cdot(4^{L+m}\cdot |\dom|^{L+m+1})^m+1$.
Hence we can define $g(L,l+1)=g(L,l)^{g(L,l)-1}\cdot(4^{L+g(L,l)}\cdot |\dom|^{L+g(L,l)+1})^{g(L,l)}+1$.
\end{proof}
By setting $h(L)=g(L,L)$ for each $L\in \Nat$, we immediatedly obtain the compression procedure:
\fi
\begin{corollary}
\label{cor: kernel}
There exists a computable function $h:\Nat \to \Nat$ such that every instance $\III$ of $\conv$ or $\reach$ with $\td(G)=L$ can be transformed in polynomial time into an equivalent instance $\III'$ of the same problem with $|V(G')|\le h(L)$.\end{corollary}

We are ready to prove the main theorem of the section:
\thmtreedepth*
\begin{proof}
Given an instance of $\conv$ or $\reach$ with network $G$ of treedepth $L$, we apply Corollary \ref{cor: kernel} to transform it into an equivalent instance where $G'$ has at most $h(L)$ nodes. Then $G'$ has at most $|\dom|^{h(L)}$ possible configurations. Therefore it suffices to simulate the first $|\dom|^{h(L)}$ time steps of the reduced \SDSys to solve $\conv$ or $\reach$.
\end{proof}

\section{Restricting the In-Degree}
In this final section, we turn our attention to \allconv. In particular, while one cannot hope to extend \Cref{thm: treedepth} to the \allconv\ problem due to known lower bounds~\cite{RosenkrantzMRS21}, one can observe that the reduction used there requires nodes with high in-degree. Here, we show that when we restrict the inputs by including the in-degree as a parameter in addition to treedepth, the problem becomes fixed-parameter tractable.

Let us start by showing that \allconv\ can be solved efficiently for networks without long directed paths and nodes of large in-degrees. In fact, the same argument also allows us to obtain a more efficient algorithm for \conv\ in this setting.
\begin{lemma}
$\allconv$ (or \conv) can be solved in time $|\dom|^{2(pd^p+1)} \cdot \bigoh(n^3)$ (or $|\dom|^{pd^p+1} \cdot \bigoh(n^3)$, respectively), where:
\begin{itemize}
\item $p$ is the maximum length of a directed path in the network,
\item $d$ is the maximum in-degree of the input network, and
\item $n$ is the number of nodes in the network.
\end{itemize}
\end{lemma}
\begin{proof}
For a node $v\in V(G)$, we denote by $X_v$ the set of all $u\in V(G)$ such that $G$ contains a directed path from $u$ to $v$. Observe that to solve the instance $\III=(S,x)$ of \conv\ (or $\III=S$ of \allconv), it is sufficient to solve its restriction to every set $X_v$ (denoted $\III_v=(S_v,x_v)$ or $\III_v=S_v$ correspondingly). Let $d$ and $p$ be the maximum in-degree and length of a simple directed path in $G$ respectively, then each $X_v$ contains at most $pd^p+1$ elements. Therefore $S_v$ can have at most $|\dom|^{pd^p+1}$ different configurations. For \conv, we start from $x_v$, simulate $|\dom|^{pd^p+1}$ time steps and check whether the resulting configuration is a fixed point. We return "Yes" if and only if every $(S_v, x_v)$ reaches a fixed point. In case of \allconv, we proceed similarly, but for every $S_v$ at first branch over at most $|\dom|^{pd^p+1}$ possible starting configurations. Since the number of sets $X_v$ is $\bigoh(n)$ and the simulation of one step requires time of at most $\bigoh(n^2)$, we get the time bounds of $|\dom|^{pd^p+1} \cdot \bigoh(n^3)$ and $|\dom|^{2(pd^p+1)} \cdot \bigoh(n^3)$ for \conv\ and \allconv\ correspondingly.
\end{proof}

As an immediate corollary, we have:
\thmtreedepthindeg*

\iflong
However, bounding only the in-degrees of nodes is not sufficient to achieve tractability of the problem:
\fi
\ifshort
However, bounding only the in-degrees of nodes is not sufficient to achieve tractability of the problem. Indeed, by reducing from 3-\textsc{unSAT} we obtain:
\fi

\iflong
\begin{theorem}
\fi
\ifshort
\begin{theorem}
\fi
\allconv\ is co-\NP-hard even if $G$ is a DAG with maximum in-degree of $3$.
\end{theorem}
\iflong
\begin{proof}
It is sufficient to prove the statement for $\dom=\{0,1\}$. In this case there is a natural correspondence between the elements 0 and 1 of the domain and the boolean values 0 and 1. We reduce from 3-\textsc{unSAT} problem where each clause has size of three as follows. Let $Y$ and $C$ be the sets of variables and clauses correspondingly. For every variable $y\in Y$, we create a node $v_y$ which will be a source in $G$. We define the local functions as constants, i.e., $f_{v_y}(0)=0$ and $f_{v_y}(1)=1$. Then, to model clause $c$ containing variables $y_1$, $y_2$ and $y_3$, we create a node $w_c$ and add arcs from $y_1$, $y_2$ and $y_3$ to $w_c$. We define $f_{w_c}$ so that the state of $w_c$ corresponds to the boolean value of $c$ when the boolean values of $y_1$, $y_2$ and $y_3$ are determined by the states of $v_{y_1}$, $v_{y_2}$ and $v_{y_3}$. In particular, the state of every $w_c$ stabilises after at most one step. Let $C=\{c_i:i\in[m]\}$. To model the conjunction of clauses, we create new nodes $u_i$, $i\in [m-1]$, add arcs from $c_1$ and $c_2$ to $u_1$ and then from $u_i$ and $c_{i+2}$ to $u_{i+1}$ for every $i\in [m-2]$. The local functions of each $u_i$ acts as a logical ``and'' of the states of its in-neighbours. Then $u_i$ stabilises after at most $i+1$ steps in state with corresponds to a conjunction of first $i+1$ clauses. We add an auxiliary node $v_0$ with the only in-neighbour $u_{m-1}$ and define it's local function so that the state of $v_0$ remains constant if and only if $u_{m-1}$ is in state 0 (and alternates otherwise). Notice that the state of $v_0$ stabilises if and only if the initial states of $v_y$, $y\in Y$, do not form a satisfying assignment. 
\end{proof}
\fi
\section{Concluding Remarks}

Our results shed new light on the complexity of the three most fundamental problems on synchronous dynamic systems. They also identify two of these---\reach\ and \conv---as new members of a rather select club of problems with a significant complexity gap between parameterizing by treewidth and by treedepth. It is perhaps noteworthy that the few known examples of this behavior are predominantly (albeit not exclusively~\cite{GutinJW16}) tied to problems relevant to AI research~\cite{GanianO18,GanianPSS20,GanianHO21}.

One question left open for future work is the exact complexity classification of \allconv\ on networks of bounded treedepth or treewidth. Indeed, while previous work~\cite{RosenkrantzMRS21} shows that the problem is co\NP-complete on DAGs, it is not clear why the problem should be included in co\NP\ on general networks (in particular, while convergence from a fixed starting state is polynomial-time checkable on DAGs, it is \PSPACE-complete on general networks). Another question that could be tackled in future work is whether Theorem~\ref{thm: treedepth_indeg} can be generalized to treewidth instead of treedepth.

\section{Acknowledgements} 
Robert Ganian and Viktoriia Korchemna acknowledge support by the Austrian Science Fund (FWF, project Y1329). 
Thekla Hamm acknowledges support by the Austrian Science Fund (FWF, project J4651).
\bibliography{references}

\end{document}